\documentclass[american,aps,prx,reprint,floatfix,nofootinbib,superscriptaddress,notitlepage]{revtex4-2}
\usepackage[unicode=true,pdfusetitle, bookmarks=true,bookmarksnumbered=false,bookmarksopen=false, breaklinks=false,pdfborder={0 0 0},backref=false,colorlinks=false]{hyperref}
\hypersetup{colorlinks,linkcolor=myurlcolor,citecolor=myurlcolor,urlcolor=myurlcolor}
\usepackage{graphics,epstopdf,graphicx,amsthm,amsmath,amssymb,
braket,colortbl,color,bm,framed,mathrsfs}

\usepackage[T1]{fontenc}
\usepackage[up]{subfigure}
\usepackage{tikz}
\usepackage{algorithm}
\usepackage{algorithmic}
\usepackage{enumerate}
\usepackage{multirow}

\usepackage{tcolorbox}

\definecolor{myurlcolor}{rgb}{0,0,0.9}

\newcommand{\proj}[1]{| #1\rangle\!\langle #1 |}

\newcommand{\inner}[2]{\langle #1 , #2\rangle}

\DeclareMathOperator{\trace}{Tr}
\newcommand{\Ptr}[2]{\trace_{#1}\Pa{#2}}
\newcommand{\Tr}[1]{\Ptr{}{#1}}

\newcommand{\Pa}[1]{\left[#1\right]}

\theoremstyle{plain}
\newtheorem{thm}{Theorem}
\newtheorem{lem}[thm]{Lemma}

\newtheorem{Def}[thm]{Definition}

\newtheorem{Exam}[thm]{Example}
\newenvironment{example}%
  {\begin{Exam}\upshape}{\end{Exam}}

\newcommand*{\myproofname}{Proof}

\def\ot{\otimes}

\def\cM{\mathcal{M}}

\def\cD{\mathcal{D}}

\def\cH{\mathcal{H}}

\usepackage{ulem,xpatch}
\xpatchcmd{\sout}
  {\bgroup}
  {\bgroup}
  {}{}
\newcommand{\be}{\begin{equation}}
\newcommand{\ee}{\end{equation}}
\newcommand{\beq}{\begin{eqnarray}}
\newcommand{\eeq}{\end{eqnarray}}

\DeclareMathAlphabet{\mathcal}{OMS}{cmsy}{m}{n}

\makeatother

\begin{document}
\title{Magic Class and the Convolution Group}

\author{Kaifeng Bu}
\email{kfbu@fas.harvard.edu}
\affiliation{\it Department of Physics, Harvard University, Cambridge, MA 02138, USA
}

\author{Arthur Jaffe}
\email{Arthur\_Jaffe@harvard.edu}
\affiliation{\it Department of Physics, Harvard University, Cambridge, MA 02138, USA
}
 \affiliation{\it Department of Mathematics, Harvard University, Cambridge, MA 02138, USA}

\author{Zixia Wei}
\email{zixiawei@fas.harvard.edu}
\affiliation{\it Department of Physics, Harvard University, Cambridge, MA 02138, USA
}
\affiliation{\it Society of Fellows, Harvard University, Cambridge, MA 02138, USA}

\begin{abstract}
The classification of many-body quantum states plays a fundamental role in 
the study of quantum phases of matter. In this work, we propose an approach to classify quantum states by 
 introducing the concept of magic class. 
In addition, we introduce an efficient coarse-graining procedure to extract the magic feature of states, which we call the ``convolution group (CG).''   We classify quantum states into different magic classes using the fixed points of the CG and circuit equivalence. We also show that magic classes can be characterized by symmetries and the quantum entropy of the CG fixed points. 
Finally, we discuss
the connection between the CG and the renormalization group. 
These results may provide new insight into the study of the state classification and quantum phases of matter.
\end{abstract}

\maketitle

\newpage

\section{Introduction}
Classifying quantum states according to certain features is essential in understanding quantum many-body systems with large degrees of freedom. One famous example is the classification of quantum phases of matter. 
Quantum phase~\cite{carr2010,sachdev_2011} is a foundational concept,
particularly in the realm of condensed matter physics and quantum field theory. 
The exploration of quantum phases involves categorizing distinct states in quantum many-body systems. These phases are characterized by symmetries~\cite{Landau37}, correlations~\cite{CGW10}, or other physical properties. 

Given a many-body quantum state, a practical way to tell its quantum phase is by applying the renormalization group (RG) method ~\cite{Kadanoff66,Wilson74,Fisher98}. Besides the
classifications of quantum phases of matter~\cite{CGW10,Yoshida11}, the RG has many applications including the
critical phenomena~\cite{Wilson74}, the efficient simulation of quantum-many body systems~\cite{White92,Vidal05}, neural networks~\cite{mehta2014exact,Ringel18,GordonPRL21,Kline22,CotlerPRD23,Cotler23} and quantum error correction~\cite{Almheiri_2015,Pastawski_2015,Yang_2016,Cong19, cong2022enhancing,lake2022exact}.
RG is an important tool for understanding the macroscopic behavior of systems and elucidating the diverse phases of matter in quantum many-body systems. 

In this paper, we focus on one important quantum feature called magic (or non-stabilizerness)~\cite{BravyiPRA05}.
Stabilizer states were first introduced by  Gottesman~\cite{Gottesman97} in the study of quantum error correction codes (also known as stabilizer codes). Moreover, the Gottesman-Knill theorem shows  that quantum circuits with stabilizer input states, Clifford unitaries, and Pauli measurements can be simulated efficiently on a
  classical computer~\cite{gottesman1998heisenberg}. This suggests that nonstabilizerness is key for quantum computational advantages. The term ``magic'' was introduced by Bravyi and Kitaev~\cite{BravyiPRA05} to describe nonstabilizerness. The extension of the Gottesman-Knill theorem was further studied beyond stabilizer circuits~\cite{BravyiPRL16,BravyiPRX16,bravyi2019simulation,BeverlandQST20,SeddonPRXQ21, bu2022classical,gao2018efficient,Bu19,UmeshSTOC23}.
  Furthermore, a resource theory of magic was developed~\cite{Veitch12mag,Veitch14,Bucomplexity22,BuPRA19_stat,RallPRA19,RyujiPRL19,WangNJP19,LeonePRL22, Chen22,
Haug23,HaugPRB23, WangPRA23,haug2023efficient,HL2023stabilizer}.

Recently, magic, as a quantum feature, has also been used to study quantum phases of matter 
and phase transitions~\cite{Ellison2021symmetryprotected,LiuPRXQ22,WhitePRB21,fux2023,niroula2023,leone2023phase,ZhouSciPhy20,True2022transitionsin,HammaPRA22,tarabunga2023critical}. 
  For example,  magic has been introduced to 
  study the complexity of the quantum phases of matter, e.g., symmetry-protected topological (SPT) phases. Furthermore, new concepts like symmetry-protected magic were introduced~\cite{Ellison2021symmetryprotected,LiuPRXQ22}. Moreover, 
  the magic of the critical point of a
  Potts model  has  been studied to show that  the conformal field theory
is magical \cite{WhitePRB21}. 

In this work, we propose a systematic way to classify quantum states into different ``magic classes'' according to their magic features. After presenting one approach based on the equivalence relation induced by the Clifford circuits, we introduce a coarse-graining map, called the ``convolution group (CG)'', as a practical way to classify the magic classes in terms of the fixed points of CG and establish a connection between the two approaches. 
The construction of CG is based on the quantum convolution~\cite{BGJ23a,BGJ23b,BGJ23c}, which has been used to prove the quantum central limit theorem~\cite{BGJ23a,BGJ23b,BGJ23c}, the extremality of stabilizer states~\cite{BJ24a}, and quantum inverse sumset theorem~\cite{BGJ24a}. We show that the magic class can be characterized by symmetries and the quantum entropy of the CG fixed points. We will also show that the CG is efficient in the since that it converges rapidly to the fixed points. The mathematical structure of the magic class and the CG is inspired by that of quantum phases and RG, and we will present a short review of the latter (see the example of 1D transverse-field Ising in the Preliminary), and comment on their connections with the former throughout the paper.

\section{Preliminaries}
We focus on an $n$-qudit system denoted as $\mathcal{H}^{\ot n}$, where $\mathcal{H}$ is a $d$-dimensional Hilbert space. For simplicity, we take $d$ to be a prime number. We fix an orthonormal basis $\{\ket{i}\}_{i \in \mathbb{Z}_d}$ for $\mathcal{H}$, where $\mathbb{Z}_d$ is the cyclic group of order $d$. The orthonormal basis for $\mathcal{H}^{\ot n}$ is denoted as $\set{\ket{\vec i} \equiv | i_1 \rangle \otimes \cdots \otimes | i_n \rangle}$, known as the computational basis. We use $\cD(\cH^{\otimes n})$ to denote the set of $n$-qudit states.
In the following, we denote a vector state by the  Dirac notation $\ket{\psi}$, and the corresponding pure-state density matrix by  
$\psi$.

The single-qudit Pauli operators $X$ and $Z$ are defined as 
$
    X: \ket{k} \mapsto \ket{k+1}$, $
    Z: \ket{k} \mapsto \omega^k_d \ket{k}
$, for all $k\in \mathbb{Z}_d$, where $\omega_d=e^{i\frac{2\pi}{d}}$. A generalized Pauli matrix is labeled by $p,q\in \mathbb{Z}_d$ and defined as 
\begin{align}                      
w(p,q)=\xi^{-pq}\, Z^pX^q, 
\end{align}
where $\xi=i$ for $d=2$, and $\xi=\omega^{(d+1)/2}_d$ for  odd $d$. 
The Weyl operators on the $n$-qudit system are defined as tensor products of generalized Pauli operators on each site, i.e.,
$
w(\vec p, \vec q)
=w(p_1, q_1)\ot...\ot w(p_n, q_n)
 $, with $\vec p=(p_1,...,p_n), \vec q=(q_1,...,q_n)$.
Let us take $V^n:=\mathbb{Z}^n_d\times \mathbb{Z}^n_d$ for simplicity, where $\vec x =(\vec p, \vec q)\in V^n$.
The Weyl operators form an orthonormal basis for linear operators on $\cH^{\otimes n}$ with respect to the inner product $\inner{A}{B}=\frac{1}{d^n}\Tr{A^\dag B}$. Thus, any quantum state $\rho$ can be expressed via Weyl operators as
$\rho=\frac{1}{d^n}
\sum_{\vec x\in V^n}
\Xi_{\rho}(\vec x)w(\vec x)$, where the characteristic function is
\begin{align}
\Xi_{\rho}(\vec x):=\Tr{\rho w(-\vec x)}.
\end{align}
Note that the process of computing the characteristic functions can be regarded as a quantum Fourier transform on qudits, which has lots of applications, including quantum Boolean functions~\cite{montanaro2010quantum}, quantum circuit complexity~\cite{Bucomplexity22}, quantum scrambling~\cite{GBJPNAS23}, classical simulation of quantum circuits~\cite{gao2018efficient,Bu19,UmeshSTOC23},  the generalization capacity 
 of quantum machine learning~\cite{BuPRA22_stat,BuQST23_stat}, and 
  quantum state 
 tomography~\cite{Bunpj22}.

{\bf Clifford unitaries and stabilizer states.}
 The Clifford unitaries on $n$ qudits are the unitaries that map Pauli operators to Pauli operators.
Stabilizer states
are pure states of the form $U_{C}\ket{0}^{\ot n}$, where $U_{C}$ is a Clifford unitary.
In the literature, 
 a convex combination of pure stabilizer states is also referred to as a stabilizer state. 

{\bf Quantum phase of matter.}
There are several ways to characterize phases in quantum many-body systems. To address their connections and differences, let us take the 1D transverse-field Ising (TFI) model as an example. The Hamiltonian of the 1D TFI model defined on a spin chain with $N$ qubits is given by 
\begin{align}
    H = -\sum_{i=1}^{N} Z_i Z_{i+1} - g\sum_{i=1}^{N} X_i, 
\end{align}
where $Z_i,X_i$ are the Pauli operators on the $i$-th site. Here we take the periodic boundary condition and identify $i=1$ with $i=N+1$. 
The ground state, denoted as $\ket{\Omega_g}$, is in the ferromagnetic phase for $0<g<1$ and 
in the paramagnetic phase for $g>1$, with $g=1$ indicating the phase transition point.

1. \textit{Fixed points under RG flow:} Quantum phases are characterized by fixed points under the RG flow. The ferromagnetic and paramagnetic phases correspond to stable fixed points 
$g\to 0$, 
and 
$g\to \infty$, respectively.

2. \textit{Spontaneously symmetry breaking (SSB):} The TFI model has a global 
$\mathbb{Z}_2$ symmetry generated by $\prod_{i=1}^N X_i$, preserved in the paramagnetic phase but broken in the ferromagnetic phase at the thermodynamic limit $N\rightarrow\infty$. This is reflected  by the ground state degeneracy at 
$g=0$, even before taking the thermodynamic limit. 

3.\textit{ Entanglement entropy at RG fixed points:} The entanglement entropy at 
$g\rightarrow0$ is $\log 2$ for any spatial bipartition, while at $g\to \infty$ it is zero.  This reflects the ground state degeneracy in the thermodynamic limit.

4. \textit{Local unitary quantum circuit:} States in the same phase are connected by a finite-depth local unitary quantum circuit, preserving long-range entanglement~\cite{CGW10}. This approach is now a standard way of defining quantum phases for ground states of gapped Hamiltonians.

Note that, in the 1D TFI model, these four characterizations are equivalent due to the model's simplicity. However, in more general cases, each characterization has limitations, and their relationships are unclear. RG fixed points, SSB, and changes in entanglement entropy may not always accompany phase transitions, and finite-depth circuits may not effectively classify gapless phases. 

{\bf Our proposed classification.}
We will introduce and study the classification of quantum states according to their magic features,
inspired by these characterizations of the quantum phases. Our method is based on the key idea of the CG, which is a 
coarse-graining map that is different from the RG. We  explain  our results in detail in the following section.

\section{Main results}
We start by introducing a method to classify quantum states into different equivalence classes, determined by connectivity through Clifford circuits.

{\bf Circuit magic class:}
    Given two $n$-qudit pure states $\ket{\psi}$ and $\ket{\phi}$, we say that they are in the same circuit magic class if  there exists a Clifford unitary $U_C$ such that $U_C\ket{\psi}=\ket{\phi}$. 

We refer to this equivalence class as the ``circuit magic class'' since it is defined by using Clifford circuits, and the properties of magic remain unchanged under Clifford unitaries. Note that this definition is inspired by the gapped phase in quantum matter, where the local unitary circuit replaces the Clifford circuit.

{\bf \textit{Convolution group.---}}
Now, let us construct an iterative coarse-graining procedure to extract the information of magic. 
In the following, we restrict $d$ to be an odd prime integer for simplicity. 
These results also apply to the case of $d=2$ with a modified quantum convolution, presented in  Appendix~\ref{appen:prop_con}.
We will usethe quantum convolution between two states~\cite{BGJ23a}:
Given $s,t \in \mathbb{Z}_d$ which satisfy $s^2+t^2\equiv 1 ~({\rm mod}~d)$, the  unitary  operator  $U_{s,t}$  acting on  a $2n$-qudit system $\mathcal{H}_A\ot \mathcal{H}_B $ is
$U_{s,t}:\ket{\vec i}\ot\ket{\vec j}\to \ket{s\vec i+t\vec j  \mod d}\ot \ket{- t\vec i+s\vec j\mod d} $,
where   both $\mathcal{H}_A$ and $\mathcal{H}_B$ are $n$-qudit systems.     
The convolution of two $n$-qudit states $\rho$ and $\sigma$ is 
\begin{align}\label{eq:conv_B}
\rho \boxtimes_{s,t} \sigma = \Ptr{B}{ U_{s,t} (\rho \otimes \sigma) U^\dag_{s,t}}.
\end{align}
Here, we consider the nontrivial parameters $s,t$, i.e., neither of them is $0$ or $1$.
Some useful properties of the quantum convolution are listed in Appendix~\ref{appen:prop_con}.

Accordingly, we can define the self-convolution for any $n$-qudit state $\rho$ as follows
\begin{align}
        \boxtimes \rho = \rho \boxtimes_{s,t} \rho.
\end{align}
In general, we can define the $L$-fold self-convolution inductively as
$\boxtimes_L \rho = \boxtimes (\boxtimes_{L-1} \rho)$, where
$\boxtimes_1 \rho = \boxtimes \rho$. 
It is important to note that the self-convolution is a map from $\cD({\cH^{\otimes n}})$ to $\cD({\cH^{\otimes n}})$, which is neither linear nor irreversible. Additionally, the self-convolution satisfies associativity, and thus it generates  a semigroup, along with the identity map. With this, we utilize the self-convolution as a coarse-graining map to construct a CG:

{\bf Construction of the CG:}
    The convolution group (CG) is formed by the $L$-th self-convolution $\boxtimes_{L}$, for  $(L=1,2,\cdots)$, along with the identity map. See Fig.~\ref{fig:mag_class} for a sketch.

\begin{figure}
    \includegraphics[width=8.0cm]{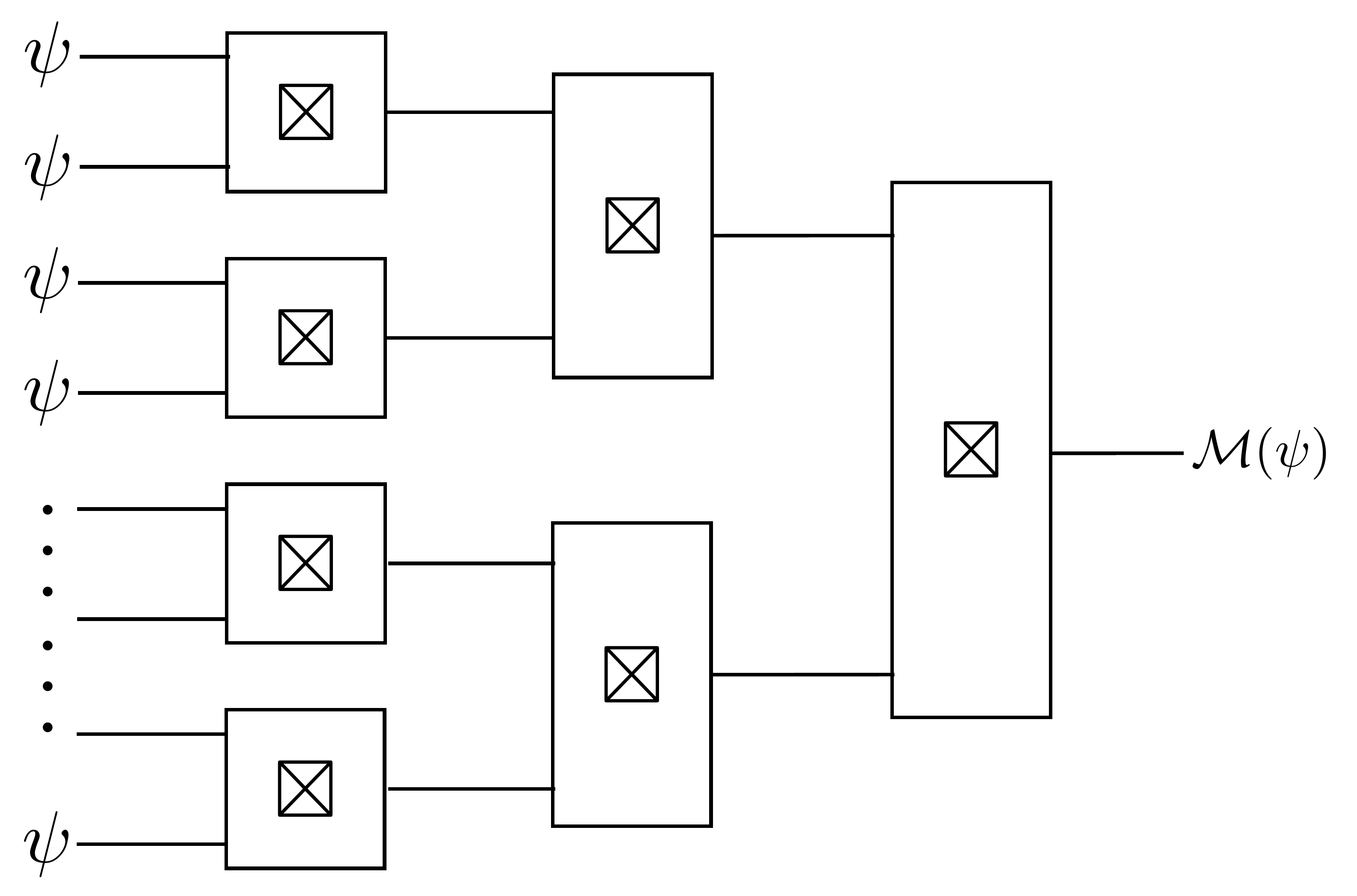}
    \caption{A diagram of the CG as an iterative coarse-graining quantum circuit. Pure $n$-qudit  states can be classified into $n+1$ CG magic classes according to the fixed point under the CG flow.}
    \label{fig:mag_class}
\end{figure}

{\bf The fixed points of the CG:}
To investigate the fixed points under the CG, it is convenient to consider them in terms of Weyl operators.
The concept of mean states~\cite{BGJ23a,BGJ23b} is crucial in this context. For an $n$-qudit state $\rho$, its mean state $\mathcal{M}(\rho)$ is defined by the characteristic function as follows:
\begin{align}\label{0109shi6}
\Xi_{\mathcal{M}(\rho)}(\vec x) :=
\left\{
\begin{aligned}
&\Xi_\rho ( \vec x) , && |\Xi_\rho ( \vec x)|=1,\\
& 0 , && |\Xi_\rho (  \vec x)|<1.
\end{aligned}
\right.
\end{align}
Moreover, the set of Weyl operators $\set{w(\vec x):\Xi_{\mathcal{M}(\rho)}(\vec x)\neq 0}$ 
form an abelian subgroup( see Lemma 12 in ~\cite{BGJ23b}), which we call the stabilizer group of $\mathcal{M}(\rho)$, and denote  as $G_{\rho}$ for simplicity.

In fact, the mean state is the fixed point of the CG. Because,
the quantum convolution of two states yields multiplication of  the quantum Fourier coefficients, i.e.  $\Xi_{\rho\boxtimes_{s,t}\sigma}(\vec x)=\Xi_{\rho}(s\vec x)\Xi_{\sigma}(t\vec x)$. (See Appendix~\ref{appen:prop_con}.) As a consequence, $|\Xi_{\boxtimes_L\rho}(\vec x)|$ will converge to 0 as $L\to \infty$, for any $\vec x\in V^n$ for which $|\Xi_\rho (  \vec x)|<1$.  Hence, the fixed points of the CG are equivalent to a mean state, up to  Weyl operators; this is the ``quantum central limit theorem'' for  DV systems~\cite{BGJ23a, BGJ23b, BGJ23c}.   Thus  we can represent the fixed points of the CG in terms of  mean states.

{\bf Remark:}
While the CG can be easily performed iteratively on the computational basis, it is convenient to analyze fixed points analytically by the quantum Fourier transform from the computational basis to the Weyl  operator basis. RG has a similar mathematical feature, where the real-space RG can be easily performed iteratively, but for analytical analysis of fixed points, it is convenient to Fourier transform to the momentum space and perform the momentum space RG.

Now we can
classify the states into different equivalence classes based on the fixed points of the CG, which we term
 CG magic classes.

 {\bf CG magic class:}
    We say that two $n$-qudit pure states $\ket{\psi}$ and $\ket{\phi}$  are in the same CG magic class if  their fixed points are equivalent up to a Clifford unitary.

We find that the circuit magic class and CG magic class have the following relation with the proof given in 
Appendix~\ref{appen:prop_con}.

\begin{thm} [\bf Relating the magic classes]\label{thm:equiv}
    Any two pure, $n$-qudit states $\ket{\psi}$ and $\ket{\phi}$  in the same circuit magic class are in the same CG magic class. 
\end{thm}

Now, we have established a connection between our definition of the circuit magic class and the CG magic class, demonstrating that the 
CG magic class provides a distinct yet non-trivial classification of pure $n$-qudit states. This situation is similar to the classification of quantum phases, where phases can be classified by either the fixed points of the RG or equivalence classes with respect to local unitary circuits. Although these two classifications are not strictly equivalent, they often coincide in practice.

{\bf \textit{Classification of CG magic classes.---}}
 Next, let us explore the classification of CG magic classes. Note that there exist $n+1$ distinct equivalence classes based on the size of the corresponding stabilizer group of the fixed points.
 The size of the stabilizer group can take values of $d^k$, where $k$ is any integer satisfying $0\leq k\leq n$. Hence, we can say a state $\psi$ is in the $k$-th CG magic class when the stabilizer group size is $d^{n-k}$. In fact, the integer $k$ directly reflects the number of non-Clifford gates required in the state preparation, as illustrated in the following example.

\begin{example}
   Let us consider an $n$-qudit system and a family of  quantum states 
as follows
\begin{eqnarray}
    \ket{\psi_k}
    =U_{C}\ket{\text{magic}}^{\ot k}\ot \ket{0}^{\ot n-k},
\end{eqnarray}
where $U_{C}$ is a Clifford unitary,  $ \ket{\text{magic}}$ is a single-qudit state chosen as $(\ket{0}+\ket{1})/\sqrt{2}$ for the local dimension $d$ being odd , or 
$(\ket{0}+e^{i\pi/4}\ket{1})/\sqrt{2}$ for $d=2$.
These states play an important role in the magic-based quantum computation~\cite{jozsa2014classical,koh2015further,BravyiPRX16,Yoganathan19}.
Computing the corresponding mean state 
$\mathcal{M}(\psi)$, one can find that $\psi_k$ is in the $k$-th magic class.

\end{example}

{\bf\textit{ Symmetry characterization.---}}
Let us then see that the CG magic class can also be characterized by the symmetries. Quantum phase transitions often (though not always) come along with spontaneous symmetry breaking, where different phases have distinct symmetries. In the 1D TFI model we have reviewed, in the thermodynamic limit, the paramagnetic phase has the $\mathbb{Z}_2$ symmetry while the ferromagnetic phase does not. These features manifest at the RG fixed points. Here, 
we find that our CG magic class also presents a symmetry characterization as follows.

\begin{thm}[\bf Symmetry characterization of CG magic class]\label{thm:sym}
Any $n$-qudit pure state $\ket{\psi}$  is in the $k$-th CG magic class, if and only if the number of the Weyl operators $w(\vec a)$ such that $[\mathcal{M}(\psi), w(\vec a)] = 0$ is $d^{n+k}$.

Equivalently, $\psi$ is in the $k$-th CG magic class if and only if the number of the Weyl operators $w(\vec a)$ such that: 1. $[\mathcal{M}(\psi), w(\vec a)]=0$ and 2. $w(\vec a)$ is not in the stabilizer group 
of $\mathcal{M}(\psi)$,
is $d^{2k}$.
\end{thm}
The proof is presented in Appendix~\ref{appen:prop_con}.
Based on this result, the CG magic classes can be recognized according to the number of symmetries with respect to Weyl operators at the fixed points of the CG.

{\bf\textit{ Entropy characterization---}}Now that we have established that the CG can categorize $n$-qudit pure states into $n+1$ CG magic class, let's explore this classification based on the quantum entropy of the fixed points. Given that mean states $\cM(\rho)$ exhibit a flat spectrum, we have the following result, providing an entropy-based characterization of the CG magic class.

\begin{thm}[\bf Entropy characterization of CG magic class]\label{thm:entro}
    Any $n$-qudit pure state $\ket{\psi}$ is in the $k$-th CG magic class, if and only if the von Neumann entropy of the fixed point $\mathcal{M}(\psi)$ is $k\log d$, i.e., $S(\mathcal{M}(\psi))=k\log d$.
\end{thm}
The proof is presented in Appendix~\ref{appen:prop_con}.
In practical scenarios, obtaining the fixed point exactly is impractical due to the need for an infinite number of self-convolution iterations. After recognizing that the entropy of fixed points varies by at least $\log d$ across different classes, our goal is to ensure that the entropy of the output state after the $L$-th self-convolution closely approximates that of the fixed point, within a margin of $\frac{1}{2}\log d$. Utilizing the entropic quantum central limit theorem~\cite{BGJ24a}, we can establish the following constraint on the entropy difference for the  $L$-th self-convolution
\begin{eqnarray}
\nonumber && |S(\mathcal{M}(\psi))- S(\boxtimes_{L}\psi)|\\
   &\leq& \log\left[1+(1-MG(\psi))^{2^{L+1}-2}e^{S(\mathcal{M}(\psi)}
   \right]\\
  \nonumber &\approx& \log\left[1+e^{-MG(\psi)(2^{L+1}-2)}e^{S(\mathcal{M}(\psi)} \right],
\end{eqnarray}
where the magic gap is defined as
$
MG(\rho)=1-\max_{\vec x\in  \text{Supp}(\Xi_{\rho}): |\Xi_{\rho}(\vec x)|\neq 1}|\Xi_{\rho}(\vec x)|
$~\cite{BGJ23a,BGJ23b}.
If $\set{\vec x\in  \text{Supp}(\Xi_{\rho}): |\Xi_{\rho}(\vec x)|\neq 1}=\emptyset$, define  $MG(\rho)=0$, i.e., there is no gap on the support of the characteristic function. As a result, 
$L=O\left(\log\left(n\frac{\log d}{\log (1-MG(\psi))^{-1}}\right)\right)$ is sufficient for telling the CG magic class of a given pure state.

Note that the entropy characterization of CG magic classes is inspired by the entanglement entropy characterization of quantum phases of matter. For instance, in the finite-size 1D TFI model, the RG fixed point of the ferromagnetic phase exhibits an entanglement entropy of $\log 2$, while that of the paramagnetic phase is zero. These entropies correspond to different ground state degeneracies at the thermodynamic limit and characterize their respective phases.

\section{Conclusion and Discussion}

In this work, we propose a framework to classify quantum states according to their magic features based on two approaches: the circuit magic class and the CG magic class, and build a connection between them. In this procedure, we introduce and construct CG as an iterative coarse-graining map
to implement the classification, which turns out to be efficient due to its rapid convergence behavior. We also discuss two other characterizations of magic classes, including the symmetries and quantum entropy of the CG fixed points. This coarse-graining method to extract magic indicates that the 
magic is a macroscopic feature.

In addition to the results in this work, it is also intriguing to investigate the behavior of states in 
typical quantum many-body systems under the CG flow. Exploring this direction may provide some new insights into the classification of quantum phases of matter.

\section*{Acknowledgements}

We are grateful to Roy Garcia, Bertrand Halperin, Yichen Hu, Zhu-Xi Luo, Tomohiro Soejima, Jonathan Sorce, Brian Swingle, Jie Wang, Yasushi Yoneta and Carolyn Zhang for useful discussions. KB and AJ are supported in part by the ARO Grant W911NF-19-1-0302 and the ARO
MURI Grant W911NF-20-1-0082. ZW is supported by the Society of Fellows at Harvard University.

\bibliography{reference}{}
\clearpage
\newpage
\onecolumngrid
\appendix
\section{Quantum convolutions on DV systems}\label{appen:prop_con}
In a series of work~\cite{BGJ23a,BGJ23b,BGJ23c,BGJ24a,BJ24a}, the framework of quantum convolutions on DV quantum systems have been built. Here, we recall some basic properties of quantum convolutions.  First, the quantum convolution $\boxtimes_{s,t}$ defined in \eqref{eq:conv_B} on qudit-system for odd prime $d$ have the following properties.

\begin{lem}[\cite{BGJ23a,BGJ23b}]\label{lem:key_tech}
The quantum convolution $\boxtimes_{s,t}$ satisfies the following properties:
     \begin{enumerate}
    \item {\bf Convolution-multiplication duality:} $$\Xi_{\rho\boxtimes_{s,t}\sigma}(\vec x)=\Xi_{\rho}(s\vec x)\Xi_{\sigma}(t\vec x),$$ for any $\vec x\in V^n\times V^n$.
    \item  {\bf Convolutional stability:} If both $\rho$ and $\sigma$ are stabilizer states,  then 
$\rho\boxtimes_{s,t}\sigma$ is still a stabilizer state.
    \item {\bf Quantum central limit theorem:} The iterated convolution $\boxtimes^N\rho$   converges to the stabilizer state $\mathcal{M}(\rho)$ as $N\to \infty$.  

\item {\bf Quantum maximal entropy principle:} $S(\rho)\leq S(\mathcal{M}(\rho))$.

\item{\bf Commutativity with Clifford unitaries: }for any Clifford unitary $U$, there exists some Clifford unitary $V$ such
that $(U\rho U^\dag)\boxtimes (U\sigma U^\dag)=V(\rho\boxtimes \sigma)V^\dag$ for any input states $\rho$ and $\sigma$.

\end{enumerate}
\end{lem}

To the consider $d=2$ case, 
we need to change
the definition of quantum convolution to the one  used in~\cite{BGJ23c}.

\begin{Def}[\bf Key Unitary for qubit systems]\label{def:key_U}
The key unitary $V$ for  $3$ quantum systems, with each system containing $n$ qubits, is 
\begin{align}\label{eq:con_cir}
V:
=U_{1,n+1,2n+1}\ot U_{2,n+2,2n+2}\ot ...\ot U_{n, 2n, 3n}.
\end{align}
Here  $U$ is a $3$-qubit unitary constructed using CNOT gates:
\begin{eqnarray}
U:=\left(\prod^3_{j=2}CNOT_{j\to 1}\right)\left(\prod^3_{i=2}CNOT_{1\to i}\right),
\end{eqnarray}
and 
$
CNOT_{2\to 1}\ket{x}\ket{y}=\ket{x+y}\ket{y}
$ for any $x,y\in\mathbb{Z}_2$. 
\end{Def}

\begin{Def}[\bf Convolution of three states  for qubit systems]\label{Def:conv_qubit}
Given $K$ states $\rho_1,\rho_2,\rho_3$, each with $n$-qubits, the multiple convolution $\boxtimes_3$ 
of $\rho_1,\rho_2,\rho_3$ maps to an $n$-qubit state: 
\begin{align}\label{eq:qub_con}
\boxtimes_{3}(\rho_1,\rho_2,\rho_3)=\boxtimes_3(\ot^3_{i=1}\rho_i)=\Ptr{1^c}{V(\ot^3_{i=1}\rho_i )V^\dag}\;.
\end{align}
Here $V$ is the key unitary in Definition~\ref{def:key_U}, and
$\Ptr{1^c}{\cdot}$ denotes the partial trace taken on the subsystem $2, 3$, i.e., 
 $\Ptr{1^c}{\cdot}=\Ptr{2,3}{\cdot}$.
\end{Def}

Note that the quantum convolution on qubit-system also satisfy the 
properties in Lemma \ref{lem:key_tech}. The details can be found in 
\cite{BGJ23c}.

\section{The proofs of the Theorems}

\begin{thm} [Restatement of Theorem~\ref{thm:equiv}]
    For any pure $n$-qudit states $\ket{\psi}$ and $\ket{\phi}$, if they are in the same circuit magic phases, they are in the same CG magic classes. 
\end{thm}
\begin{proof}
Based on the commutativity of quantum convolution with Clifford unitaries~\cite{BGJ23b,BGJ23c}, i.e, 
for any Clifford unitary $U$, there exists some 
Clifford unitary $V$ such that 
\begin{eqnarray}
    \boxtimes(U\rho U^\dag)=V(\boxtimes \rho)V^\dag
\end{eqnarray}
Hence, for $\boxtimes_L$, there also exists a some Clifford unitary $V_L$ such
that
\begin{eqnarray}
    \boxtimes_L (U\rho U^\dag)=V_L(\boxtimes_L \rho)V^\dag_L
\end{eqnarray}
Hence, the fixed point of the RG on $\psi$ and $\phi$, i..e, the mean states $\mathcal{M}(\psi)$ and $\mathcal{M}(\phi)$,  are connected by a Clifford unitary. 
Hence, they are in the same CG magic classes.

\end{proof}

\begin{thm}[Restatement of Theorem~\ref{thm:sym}]
    For any $n$-qudit pure state $\ket{\psi}$, it is in the $k$-th CG magic class if and only if the number of the Weyl operators $w(\vec a)$ such that $[\mathcal{M}(\psi), w(\vec a)] = 0$ is $d^{n+k}$.

Equivalently, $\ket{\psi}$ is in the $k$-th CG magic class if and only if the number of the Weyl operators $w(\vec a)$ such that: 1. $[\mathcal{M}(\psi), w(\vec a)]=0$ and 2. $w(\vec a)$ is not in the stabilizer group 
of $\mathcal{M}(\psi)$,
is $d^{2k}$.
\end{thm}
\begin{proof}

First, $\psi$ is the in the $k$-th CG magic class  iff the set of Weyl operators $G_{\psi}=\set{w(\vec x):\Xi_{\mathcal{M}(\psi)}(\vec x)\neq 0}$ 
form an abelian subgroup with the size being $d^{n-k}$ for some integer $k\leq n$, by the Lemma 12 in ~\cite{BGJ23b}. 
That is, $G_{\psi}$ can be generated by $k$ Weyl operators, labeled as $g_1,...,g_{n-k}$, and the state $\mathcal{M}(\psi)$ can be written 
as $\mathcal{M}(\psi)=\frac{1}{d^n}\Pi_{i\leq n-k}(\sum_{k_i\in \mathbb{Z}_d}g^{k_i}_i)$.
    Hence, we can find some Clifford unitary $U_{cl}$ such that $U_{cl}g_iU^\dag_{cl}=Z_i$,  where $Z_i$ is the Pauli Z operator on the $i$-th qudit. 
Therefore, the number of  the Weyl operator  1. $[\mathcal{M}(\psi), g_i]=0$ and 2. $w(\vec a)$ is not in the stabilizer group 
of $\mathcal{M}(\psi)$,
is $d^{2k}$. 

\end{proof}

\begin{thm}[Restatement of Theorem~\ref{thm:entro}]
     For any pure state $\psi$, it is in the $k$-th CG magic class if and only if the von Neumann entropy of the fixed point $\mathcal{M}(\psi)$ is $k\log d$, i.e., $S(\mathcal{M}(\psi))=k\log d$.
\end{thm}
\begin{proof}
    Based on the above Theorem, $\psi$ is the in the $k$-th CG magic class if and only if $\mathcal{M}(\psi)=\frac{1}{d^n}\Pi_{i\leq n-k}(\sum_{k_i\in \mathbb{Z}_d}g^{k_i}_i)$, where the stabilizer group is generated by $g_1,...,g_{n-k}$.
    Hence, we can find some Clifford unitary $U_{cl}$ such that $U_{cl}g_iU^\dag_{cl}=Z_i$,  where $Z_i$ is the Pauli Z operator on the $i$-th qudit. Thus, 
    $U_{cl}\mathcal{M}(\psi)U^\dag_{cl}=\frac{1}{d^n}\Pi_{i\leq n-k}(\sum_{k_i\in \mathbb{Z}_d}Z^{k_i}_i)=\proj{0}^{n-k}\ot (\frac{I}{d})^{\ot k}$.
    Hence, 
    \begin{eqnarray}
        S(\mathcal{M}(\psi))
        =S(U_{cl}\mathcal{M}(\psi)U^\dag_{cl})
        =S\left(\proj{0}^{n-k}\ot (I/d)^{\ot k}\right)
        =k\log d.
    \end{eqnarray}
\end{proof}

\end{document}